\documentclass[letterpaper, 10 pt, conference]{ieeeconf}
\IEEEoverridecommandlockouts
\usepackage{amsmath, amssymb, bm, graphicx, cite, subcaption, booktabs}
\usepackage{cleveref}
\crefformat{equation}{(#2#1#3)}
\Crefname{algocfline}{Algorithm}{Algorithms}
\Crefname{algocf}{line}{lines}
\Crefname{AlgoLine}{Line}{Lines}
\crefname{AlgoLine}{line}{lines}

\usepackage{xcolor, placeins, textcomp, amsfonts}

\title{\LARGE \bf
Scalable Coordination with Chance-Constrained Correlated Equilibria\\via Reduced-Rank Structure
}

\author{Jaehan Im, David Fridovich-Keil, Ufuk Topcu
\thanks{This work was supported by the NSF under grants 2336840 and 2211548, by NASA under ULI grants 80NSSC21M0071 and 80NSSC24M0070, and by ONR under grant N00014-22-1-2703.}% <-this % stops a space
\thanks{Jaehan Im and David Fridovich-Keil are with the Department of Aerospace Engineering, The University of Texas at Austin,
        Austin, TX 78712, USA
        {\tt\small jaehan.im@utexas.edu; dfk@utexas.edu}}%
\thanks{Ufuk Topcu is with the Oden Institute for Computational Engineering and Sciences, The University of Texas at Austin,
        Austin, TX 78712, USA
        {\tt\small utopcu@utexas.edu}}%
}

\begin{document}
\bstctlcite{BSTcontrol}
\newcommand{\actSet}{\mathcal X}
\newcommand{\sys}{\text{sys}}

\newcommand{\ones}{\bm 1}
\newcommand{\reals}{{\mbox{\bf R}}}
\newcommand{\integers}{{\mbox{\bf Z}}}
\newcommand{\symm}{{\mbox{\bf S}}}  % symmetric matrices
\newcommand{\lag}{\mathcal{L}}

\newcommand{\nullspace}{{\mathcal N}}
\newcommand{\range}{{\mathcal R}}
\newcommand{\Rank}{\mathop{\bf Rank}}
\newcommand{\Tr}{\mathop{\bf Tr}}
\newcommand{\diag}{\mathop{\bf diag}}
\newcommand{\card}{\mathop{\bf card}}
\newcommand{\rank}{\mathop{\bf rank}}
\newcommand{\conv}{\mathop{\bf conv}}
\newcommand{\prox}{\bm{prox}}

\newcommand{\Expect}{\mathop{\bf E{}}}
\newcommand{\Prob}{\mathop{\bf Prob}}
\newcommand{\Co}{{\mathop {\bf Co}}} % convex hull
\newcommand{\dist}{\mathop{\bf dist{}}}
\newcommand{\argmin}{\mathop{\rm argmin}}
\newcommand{\argmax}{\mathop{\rm argmax}}
\newcommand{\epi}{\mathop{\bf epi}} % epigraph
\newcommand{\Vol}{\mathop{\bf vol}}
\newcommand{\dom}{\mathop{\bf dom}} % domain
\newcommand{\intr}{\mathop{\bf int}}
\newcommand{\sign}{\mathop{\bf sign}}
\newcommand{\norm}[1]{\left\lVert#1\right\rVert}
\newcommand{\mnorm}[1]{{\left\vert\kern-0.25ex\left\vert\kern-0.25ex\left\vert #1 
    \right\vert\kern-0.25ex\right\vert\kern-0.25ex\right\vert}}

\newtheorem{definition}{Definition} 
\newtheorem{theorem}{Theorem}
\newtheorem{lemma}{Lemma}
\newtheorem{corollary}{Corollary}
\newtheorem{remark}{Remark}
\newtheorem{proposition}{Proposition}
\newtheorem{assumption}{Assumption}
\newtheorem{example}{Example}

\newcommand{\cf}{{\it cf.}}
\newcommand{\eg}{{\it e.g.}}
\newcommand{\ie}{{\it i.e.}}
\newcommand{\etc}{{\it etc.}}

\newcommand{\ba}[2][]{\todo[color=orange!40,size=\footnotesize,#1]{[BA] #2}}

\newcommand{\fix}[1]{\textcolor{red}{#1}}

\newcommand{\bigO}{\mathcal{O}}

\newcommand{\intSet}{\mathbb{Z}}
\newcommand{\realSet}{\mathbb{R}}
\newcommand{\natSet}{\mathbb{N}}
\newcommand{\zeroSet}{\bm{0}}
\newcommand{\state}{\bm{x}}
\newcommand{\cmdh}{\bar{\bm{u}}}
\newcommand{\cmda}{\mathring{\bm{u}}}
\newcommand{\cmd}{\bm{u}}
\newcommand{\costh}{J_h}
\newcommand{\costa}{J_a}
\newcommand{\observ}{\bm{z}}

\newcommand{\circnum}[1]{%
  \raisebox{.5pt}{\textcircled{\raisebox{-.9pt}{#1}}}%
}

\maketitle
\thispagestyle{empty}
\pagestyle{empty}

\begin{abstract}
Chance-constrained correlated equilibrium enables coordination of noncooperative agents under cost uncertainty through probabilistic incentive-compatibility guarantees. However, computing such equilibria becomes intractable in large-scale systems due to the exponential growth of the joint action space. We develop an approximation method for computing chance-constrained correlated equilibria by showing that these equilibria admit a representation as convex combinations of a finite set of chance-constrained pure Nash equilibria, enabling tractable computation without solving the full correlated equilibrium program. Numerical experiments on large-scale multi-airline coordination scenarios demonstrate substantial reductions in computation time while achieving lower system delay costs compared to current operational practice. Under cost uncertainty, the proposed method consistently achieves lower deviation rate compared to the full formulation while achieving comparable coordination performance.
\end{abstract}

% \begin{keywords}
% Game theory, Optimization algorithms, Large-scale systems
% \end{keywords}

\section{Introduction}

Large-scale multi-agent systems often consist of self-interested agents 
whose decisions collectively determine system-level outcomes. 
Designing coordination mechanisms that improve system performance 
without centralized enforcement remains a fundamental challenge. 
Correlated equilibrium~\cite{aumann, aumann_2} provides 
a principled framework for such coordination by enabling a coordinator 
to recommend actions that agents have no incentive to deviate from. 

In many practical settings, however, agents' cost structures are uncertain due to modeling errors, operational variability, or private information \cite{ce_app_2, ce_review, j_vq, j_ccce}. This undermines the incentive-compatibility guarantees of nominal correlated equilibrium, as agents may benefit from deviating when realized costs differ from their nominal estimates. To address this issue, recent work has introduced the \emph{chance-constrained correlated equilibrium} (CC-CE) by requiring incentive compatibility to hold with a prescribed confidence level, providing robustness against cost uncertainty \cite{j_ccce, j_vq}.

Despite its modeling advantages, computing CC-CE remains challenging 
in large-scale systems. The joint action space grows exponentially 
with the number of agents and their action sets, rendering direct 
computation intractable. \textit{Reduced-rank correlated equilibrium} (RRCE) has been proposed to mitigate this issue by representing equilibria as convex combinations of a finite set of pure Nash equilibria \cite{rrce}. 
However, extending this approach to the chance-constrained setting 
is nontrivial, as the presence of probabilistic incentive constraints 
alters the structure of the feasible set.

We show that CC-CE admits a reduced-rank representation that enables scalable computation. Specifically, we approximate the CC-CE feasible set using convex combinations of a finite set of chance-constrained pure Nash equilibria, thereby avoiding explicit enumeration of the full joint action space. We further characterize the class of equilibria required for this representation and develop an algorithm that scales to large problem instances.

\begin{figure}[t!]
    \centering
    \includegraphics[width=0.95\linewidth]{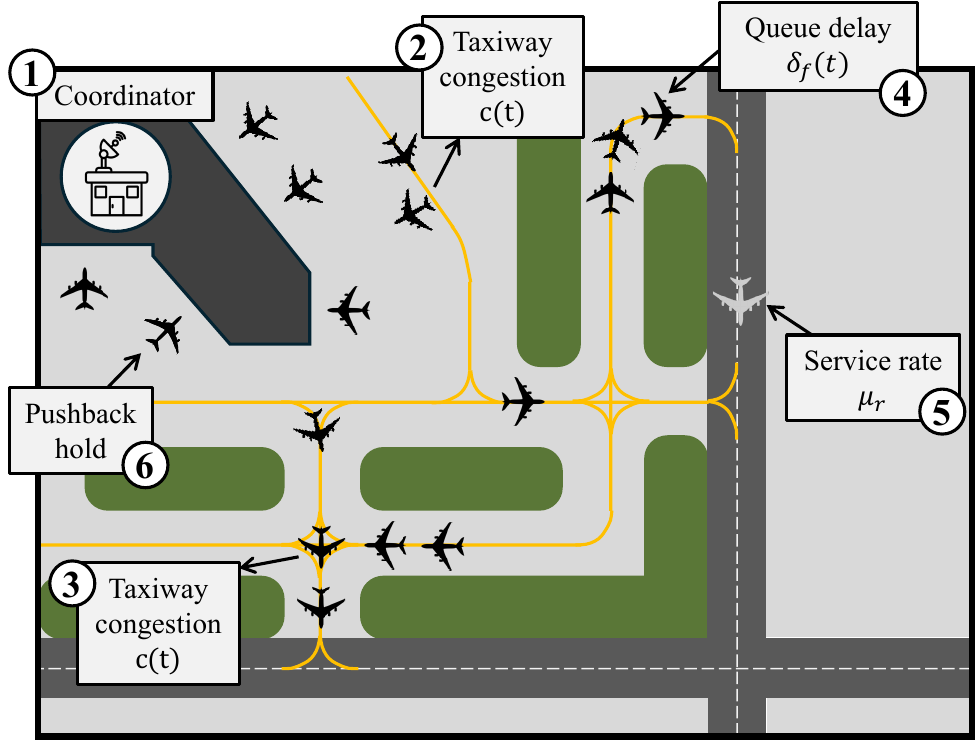}
    \caption{Illustrative example of a coordination scenario inspired by airport virtual queueing systems.
    A central coordinator recommends which aircraft to release \circnum{1} but does not enforce aircraft-level decisions directly.
    Released aircraft contribute to shared congestion $c(t)$ (\circnum{2}--\circnum{3}) and to queueing delay $\delta_f(t)$ \circnum{4}, while departures are governed by runway service rates $\mu_r$ \circnum{5}.
    Airlines retain autonomy over whether to follow the recommendations \circnum{6}, creating a noncooperative coordination problem.}
    \label{fig:concept}
\end{figure}

We demonstrate the effectiveness of the proposed approach on coordination problems motivated by virtual queueing systems, where large-scale coordination under uncertainty is required. Numerical results show that the reduced-rank method achieves substantial computational gains while maintaining comparable coordination performance.

The contributions of this paper are threefold. 
First, we establish a reduced-rank representation of chance-constrained correlated equilibria via convex combinations of chance-constrained pure Nash equilibria, enabling scalable computation. 
Second, we develop a scalable algorithm that avoids enumeration of the full joint action space. 
Third, we validate the approach on large-scale coordination instances inspired by virtual queueing systems.

\section{Related Work}

The correlated equilibrium concept~\cite{aumann, aumann_2} has been widely studied as a mechanism for coordinating noncooperative agents, as it allows a coordinator to recommend actions while preserving each agent's autonomy~\cite{rrce, ce_effective_traffic, ce_app_2, ce_error_2_exp, j_vq}. However, correlated equilibrium recommendations can be sensitive to uncertainty in the underlying payoff structure, such as estimation errors or stochastic perturbations~\cite{ce_app_2, ce_review}.

Existing approaches handling uncertainty in correlated equilibria 
include Bayes correlated equilibria~\cite{ce1}, robust and 
distributionally robust formulations~\cite{ce2,ce4,ce8,ce9}, and bounded-rationality models such as quantal 
response equilibria~\cite{qre1,qre2}. While these approaches 
offer complementary perspectives, none provides explicit probabilistic 
guarantees on incentive compatibility under stochastic cost perturbations. Recent work has introduced chance-constrained formulations of correlated equilibria to address this gap~\cite{j_ccce, j_vq}.

Computing correlated equilibria is known to be intractable as the 
joint action space grows exponentially with the number of agents~\cite{rrce}. 
Prior work has addressed this through learning dynamics~\cite{tract_1, 
tract_2} and approximation methods~\cite{tract_3}, though 
these do not allow targeted equilibrium selection. The reduced-rank 
correlated equilibrium framework~\cite{rrce} approximates the correlated  equilibrium set via convex combinations of Nash equilibria, enabling tractable computation while preserving coordination quality. However, whether such reduced-rank representations extend to uncertainty-aware equilibrium formulations remains unclear.

\section{Preliminaries}

We introduce the correlated equilibrium concept and its use 
as a coordination mechanism. These form the foundation for 
the chance-constrained extension developed in \Cref{sec:ccce}.

Consider a finite strategic game with agent set 
$\mathcal{N} = \{1, \ldots, n\}$. Each agent $i \in \mathcal{N}$ 
selects an action from a finite set $\mathcal{X}_i$, and the joint 
action space is $\mathcal{X} := \prod_{i \in \mathcal{N}} \mathcal{X}_i$. 
Let $J_i: \mathcal{X} \rightarrow \mathbb{R}$ denote the cost of agent $i$. 
We define the deviation cost for agent $i$ as
\begin{equation}
    \Delta J_i(x_i, x_i', x_{-i}) := J_i(x_i, x_{-i}) - J_i(x_i', x_{-i}),
\end{equation}
where $x_{-i} \in \prod_{j \neq i} \mathcal{X}_j$ denotes the actions 
of all agents except $i$. This quantity captures the cost change that 
agent $i$ would experience by deviating from a recommended action $x_i$ 
to an alternative $x_i'$, given that other agents play $x_{-i}$.

\begin{definition}[Correlated equilibrium~\cite{aumann, aumann_2}]
A distribution $z \in \Delta(\mathcal{X})$ is a correlated equilibrium (CE) if
\begin{equation}
    \mathbb{E}_{x_{-i} \sim z(\cdot|x_i)}
    [\Delta J_i(x_i, x_i', x_{-i})] \leq 0,
    \quad \forall i,\ \forall x_i,\ \forall x_i' \neq x_i.
    \label{eq:ce}
\end{equation}
\end{definition}

In a CE, a coordinator samples a joint action from $z$ and privately 
recommends each agent its corresponding component. Since each CE 
constraint in~\Cref{eq:ce} is linear in $z$ and $\Delta(\mathcal{X})$ 
is convex, the CE set forms a convex polytope. Intuitively, once an 
agent receives a recommendation, following it is individually rational 
as no unilateral deviation reduces its expected cost.

\subsection{Coordination via correlated equilibrium}

The goal of coordination is to select, among all feasible CEs, 
one that leads to a desirable system-level outcome. Let 
$J_\mathrm{sys}: \Delta(\mathcal{X}) \rightarrow \mathbb{R}$ denote 
a system-level cost. The coordination task reduces to the following 
equilibrium selection problem:
\begin{equation}
    \min_{z \in \Delta(\mathcal{X})}\ J_\mathrm{sys}(z) \quad 
    \text{s.t. } z \text{ satisfies~\Cref{eq:ce}}.
    \label{eq:ce_program}
\end{equation}
Since both the objective and constraints are linear in $z$, 
\Cref{eq:ce_program} is a linear program. This formulation 
enables the coordinator to steer self-interested agents toward 
system-efficient outcomes while preserving their autonomy.
\medskip
\begin{example}[Coordination via correlated equilibrium]
Consider a scenario where two drivers approach an intersection and choose either 
\emph{Go} (G) or \emph{Stop} (S).
Their costs are given by
% \small{
\begin{equation}
\begin{array}{c|cc}
 & G & S \\
\hline
G & (5,5) & (-1,1) \\
S & (1,-1) & (1,1)
\end{array}
\end{equation}
% }
\normalsize
where rows correspond to Driver~1 and columns to Driver~2.
If both go, they incur a large penalty (5,5).
If both stop, both incur a mild delay (1,1).
If exactly one goes, the driver who goes benefits ($-1$) 
while the other incurs a mild delay ($1$).
Consider a correlation device that selects $(G,S)$ and $(S,G)$ with probability $\frac{1}{2}$.
If a driver is recommended to go, following the recommendation yields cost $-1$,
whereas deviating leads to $(S,S)$ with cost $1$.
If recommended to stop, following yields cost $1$,
whereas deviating leads to $(G,G)$ with cost $5$.
Thus, conditioned on the received signal, deviation strictly increases cost.
The correlation device therefore avoids collision or inefficient mutual stopping.
\end{example}

\section{Chance-Constrained Correlated Equilibria}
\label{sec:ccce}

The correlated equilibrium formulation in~\Cref{eq:ce} assumes exact knowledge of 
agents' cost functions. In practice, the coordinator may 
not have accurate knowledge of agents' costs. We build on the chance-constrained correlated equilibrium 
(CC-CE) framework~\cite{j_vq, j_ccce}, which provides 
robust coordination guarantees under such uncertainty, 
and develop a scalable computation method for this setting.

\subsection{Deviation uncertainty model}

We model cost uncertainty as an additive perturbation on the 
deviation cost. For any deviation comparison involving agent $i$, 
the true deviation cost satisfies
\begin{equation}
    \Delta J_i(x_i, x_i', x_{-i}) = 
    \Delta \bar{J}_i(x_i, x_i', x_{-i}) + \eta_i,
    \label{eq:uncertainty}
\end{equation}
where $\Delta \bar{J}_i$ is the coordinator's nominal estimate and 
$\eta_i \sim \nu_i$ is an agent-level uncertainty term drawn from 
distribution $\nu_i$ with cumulative distribution function 
$\Phi_{\nu_i}: \mathbb{R} \rightarrow [0,1]$. The disturbance 
$\eta_i$ is assumed common across all deviation comparisons for 
agent $i$ within a given instance, capturing coherent modeling 
errors in the coordinator's cost estimate. Importantly, $\eta_i$ 
is not resampled for each action profile; rather, it represents 
a single agent-level realization that shifts all deviation margins 
for agent $i$ in a coherent way. Under this model, even when the 
nominal deviation margin is negative, uncertainty may cause the 
realized margin to become positive, leading agents to deviate 
from the recommendation.

\subsection{Chance-constrained correlated equilibrium}

To ensure that incentive compatibility holds despite cost 
uncertainty, we require the deviation condition to hold with 
a prescribed confidence level.

\begin{definition}[Chance-constrained CE~\cite{j_vq}]
\label{def:ccce}
For confidence level $\alpha \in (0,1)$, a distribution 
$z \in \Delta(\mathcal{X})$ is a chance-constrained correlated 
equilibrium (CC-CE) if
\begin{equation}
    \mathbb{P}_{\nu_i}\!\left(
    \mathbb{E}_{x_{-i} \sim z(\cdot|x_i)}
    [\Delta \bar{J}_i(x_i, x_i', x_{-i})] + \eta_i \leq 0
    \right) \geq \alpha,
    \label{eq:ccce}
\end{equation}
for all $i$, all $x_i$, and all $x_i' \neq x_i$.
\end{definition}

This condition guarantees that deviations remain unprofitable 
with probability at least $\alpha$. The chance 
constraint~\eqref{eq:ccce} admits a deterministic equivalent. 
Grouping all deviation constraints into an index 
$c = (i, x_i, x_i')$ and letting $i(c)$ denote the associated 
agent, \eqref{eq:ccce} is equivalent to
\begin{equation}
    \mathbb{E}_{x_{-i} \sim z(\cdot|x_i)}
    [\Delta \bar{J}_i(x_i, x_i', x_{-i})]
    + \Phi_{\nu_{i(c)}}^{-1}(\alpha) \leq 0, \quad \forall c,
    \label{eq:ccce_det}
\end{equation}
where $\Phi_{\nu_{i(c)}}^{-1}(\alpha)$ is the $\alpha$-quantile of 
the uncertainty distribution for agent $i(c)$. Since 
\eqref{eq:ccce_det} is affine in $z$, the CC-CE feasible set 
is a convex polytope, preserving the linear structure of the 
classical CE program~\cite{j_vq}. The confidence level $\alpha$ 
explicitly quantifies the tradeoff between robustness and 
feasibility: larger $\alpha$ yields stronger incentive guarantees 
but shrinks the feasible set by tightening each constraint by 
$\Phi_{\nu_{i(c)}}^{-1}(\alpha)$.

The coordination problem~\eqref{eq:ce_program} then extends 
naturally to the uncertainty-aware setting by replacing the 
CE constraints with~\eqref{eq:ccce_det}:
\begin{equation}
    \min_{z \in \Delta(\mathcal{X})}\ J_\mathrm{sys}(z) \quad 
    \text{s.t. } z \text{ satisfies~\eqref{eq:ccce_det}},
    \label{eq:ccce_program}
\end{equation}
which remains a linear program and provides the optimal 
incentive-compatible coordination policy with confidence $\alpha$.

\subsection{Scalable computation via reduced-rank structure}

Direct computation of CC-CE via~\Cref{eq:ccce_program} is 
intractable for large-scale problems, as the joint action space 
grows exponentially with the number of agents: a game with 
$n$ agents each having $m$ actions yields $\mathcal{O}(m^n)$ 
joint actions~\cite{rrce}. 

Reduced-rank correlated equilibrium 
(RRCE)\footnote{
The term \emph{reduced-rank} reflects that the joint distribution 
is restricted to a low-dimensional subspace spanned by a finite 
set of equilibria, rather than the full joint action space~\cite{rrce}.
} has been proposed to mitigate this issue 
in the nominal setting by representing equilibria as convex 
combinations of multiple Nash equilibria, each requiring 
$\mathcal{O}(mn)$ linear equations to compute via an interior 
point method~\cite{rrce}, an exponential reduction compared 
to the $\mathcal{O}(m^n)$ equations required for the full CE 
program. However, extending this approach to the 
chance-constrained setting is nontrivial, as the probabilistic 
incentive constraints alter the structure of the feasible set.

To address this challenge, we develop a reduced-rank formulation 
for CC-CE based on convex combinations of chance-constrained 
pure Nash equilibria, enabling tractable computation without 
enumerating the full joint action space.

To construct a reduced-rank representation, we begin by introducing a building block analogous to mixed Nash equilibria in the nominal setting.

\begin{definition}[Chance-constrained pure Nash equilibrium]
\label{def:ccpne}
A pure strategy profile $x \in \mathcal{X}$ is a 
chance-constrained pure Nash equilibrium (CC-PNE) with 
confidence level $\alpha$ if
\begin{equation}
    \mathbb{P}_{\nu_i}(\Delta J_i(x_i, x_i', x_{-i}) \leq 0) 
    \geq \alpha, \quad \forall i,\ \forall x_i' \neq x_i.
    \label{eq:ccpne}
\end{equation}
\end{definition}
\medskip

A CC-PNE can be identified by checking individual agents' deviation conditions without solving a joint distribution program, requiring only $\mathcal{O}(mn)$ equations per candidate~\cite{rrce}. 
The following results establish that convex combinations of CC-PNEs yield valid CC-CE distributions, justifying their use as building blocks for scalable coordination.

\begin{lemma}
\label{lem:ccpne_ccce}
Every CC-PNE induces a CC-CE distribution.
\end{lemma}
\begin{proof}
Let $x \in \mathcal{X}$ be a CC-PNE with confidence level 
$\alpha$. The induced distribution concentrates all mass on 
$x$, i.e., $z(x) = 1$ and $z(\tilde{x}) = 0$ for all 
$\tilde{x} \neq x$. Then for any $x_i' \neq x_i$,
\begin{align}
    &\mathbb{P}_{\nu_i}\!\left(
    \mathbb{E}_{x_{-i} \sim z(\cdot|x_i)}
    [\Delta J_i(x_i, x_i', x_{-i})] + \eta_i \leq 0
    \right) \nonumber \\
    &\quad = \mathbb{P}_{\nu_i}
    (\Delta J_i(x_i, x_i', x_{-i}) \leq 0) 
    \geq \alpha,
    \label{eq:lemma_proof}
\end{align}
where the equality follows from $z$ concentrating all mass 
on $x$, and the inequality follows from 
\Cref{def:ccpne}. This satisfies \Cref{def:ccce}.
\end{proof}

\Cref{lem:ccpne_ccce} shows that each CC-PNE is itself a 
valid CC-CE. We now extend this to convex combinations.

\begin{theorem}
\label{thm:convex_hull}
The convex hull of CC-PNE distributions is a subset of the 
CC-CE feasible set.
\end{theorem}
\begin{proof}
By \Cref{lem:ccpne_ccce}, each CC-PNE induces a CC-CE 
distribution. Since the CC-CE feasible set is a convex 
polytope by~\Cref{eq:ccce_det}, any convex combination 
of CC-PNE distributions also satisfies the CC-CE constraints.
\end{proof}

\Cref{thm:convex_hull} enables the following reduced-rank 
coordination approach. Let $\{x^{(k)}\}_{k=1}^d$ denote a 
finite set of $d$ CC-PNEs. By \Cref{thm:convex_hull}, any 
mixture $z = \sum_k \lambda_k \mathbf{1}\{x = x^{(k)}\}$ 
with $\lambda \in \Delta_d$ is a valid CC-CE. The coordinator 
then solves the reduced counterpart of \Cref{eq:ccce_program}:
\begin{equation}
    \min_{\lambda \in \Delta_d}\ 
    \sum_{k=1}^{d} \lambda_k\, J_\mathrm{sys}(x^{(k)}).
    \label{eq:rr_program}
\end{equation}
This reduces the coordination problem to two tractable 
steps: finding $d$ CC-PNEs each requiring $\mathcal{O}(mn)$ 
linear equations, followed by solving a $d$-variable linear program~\eqref{eq:rr_program}. Since $d \cdot \mathcal{O}(mn) 
\ll \mathcal{O}(m^n)$ in practice~\cite{rrce}, this approach 
achieves an exponential reduction in computational cost 
compared to the full CC-CE program.

% \begin{figure}[hbt!]
%     \centering
%     \includegraphics[width=0.9\linewidth]{figure/0_Concept.pdf}
%     \caption{Illustrative example of a coordination scenario inspired by virtual queueing systems.
%     A central coordinator recommends which aircraft to release \circnum{1} but does not enforce aircraft-level decisions directly.
%     Released aircraft contribute to shared congestion $c(t)$ (\circnum{2}--\circnum{3}) and to queueing delay $\delta_f(t)$ \circnum{4}, while departures are governed by runway service rates $\mu_r$ \circnum{5}.
%     Airlines retain autonomy over whether to follow the recommendations \circnum{6}, creating a noncooperative coordination problem.}
%     \label{fig:concept}
% \end{figure}

\section{Numerical Experiment}
\label{sec:experiment}

We validate the proposed framework on a coordination scenario inspired by airport virtual queueing systems. 
\Cref{fig:concept} provides an illustrative example of such a setting.
In VQ, airlines independently decide which specific aircraft to 
release for departure, while a central coordinator oversees 
system-level surface performance. Since aircraft-level decisions 
are private to each airline, the coordinator cannot directly 
enforce which aircraft are released. The proposed chance-constrained correlated equilibrium (CC-CE) framework 
enables the coordinator to recommend specific aircraft selections 
to each airline in a way that each airline finds individually 
rational to follow, steering system-level outcomes without 
overriding airline autonomy.

\subsection{Scenario description}

Let $R \in \mathbb{Z}_{>0}$ denote the number of departure runways 
and $\mathcal{F}$ the set of aircraft eligible for pushback. Each 
aircraft $f \in \mathcal{F}$ is associated with a designated runway 
$r(f)$ and an aircraft class (small, medium, or heavy). Let 
$\mathcal{I}$ denote the set of active airlines, and let 
$\mathcal{F}_i \subseteq \mathcal{F}$ denote the eligible aircraft 
operated by airline $i$. Each airline selects a subset of its 
eligible aircraft for pushback:
\begin{equation}
    x_i \in \mathcal{X}_i := 2^{\mathcal{F}_i},
    \label{eq:action}
\end{equation}
and the joint action $x = (x_i)_{i \in \mathcal{I}} \in \mathcal{X}$ 
determines which aircraft enter the surface queues. Let 
$a_r(x) := \sum_{f \in \bigcup_i x_i} \mathbf{1}\{r(f) = r\}$ 
denote the number of aircraft pushed back to runway $r$. Given 
initial queue lengths $q \in \mathbb{Z}_{\geq 0}^R$ and runway 
service rates $\mu_r$, the resulting queue length at runway $r$ is
\begin{equation}
    \tilde{q}_r(x) = q_r + a_r(x),
    \label{eq:queue}
\end{equation}
which induces a runway queue delay of $\delta_f(x) = 
\tilde{q}_{r(f)}(x) / \mu_{r(f)} \cdot 4$ minutes for flight $f$.

The per-flight delay cost incorporates schedule lateness 
$\ell_f \geq 0$, a taxiway congestion penalty
\begin{equation}
    c(x) = \max\{0,\, |\textstyle\bigcup_i x_i| - 4\}^2,
    \label{eq:congestion}
\end{equation}
and a quadratic penalty for heavily delayed flights, consistent 
with departure-management delay models~\cite{quadratic}:
\begin{equation}
    d_f(x) = \begin{cases}
        (\ell_f)^2 + \delta_f(x) + c(x), & \ell_f > 10, \\
        \ell_f + \delta_f(x) + c(x), & \text{otherwise.}
    \end{cases}
    \label{eq:delay}
\end{equation}
Airline $i$ minimizes its class-weighted delay
\begin{equation}
    J_i(x) = \sum_{f \in \mathcal{F}_i} \omega_{c(f)}\, d_f(x),
    \label{eq:airline_cost}
\end{equation}
where $\omega_{c(f)} \in \mathbb{R}_{\geq 0}$ is a weight 
determined by aircraft class, with 
$[\omega_\mathrm{heavy}, \omega_\mathrm{med}, \omega_\mathrm{small}] 
= [1.2, 1.0, 0.75]$. The coordinator minimizes the unweighted 
total delay
\begin{equation}
    J_\mathrm{sys}(x) = \sum_{f \in \mathcal{F}} d_f(x).
    \label{eq:sys_cost}
\end{equation}
Airlines apply class-dependent weights reflecting fleet-specific 
valuations, whereas the coordinator evaluates unweighted aggregate 
delay, consistent with standard VQ objectives~\cite{VQ_asdex, 
VQ_burgain}. Cost uncertainty is modeled as in~\Cref{eq:uncertainty}, 
where $\eta_i \sim \mathcal{N}(0, \sigma_i^2)$ captures unobserved 
variability in airline $i$'s deviation incentives, with $\sigma_i$ 
denoting the standard deviation of the perturbation.

\subsection{Experiment setup}
We conduct Monte Carlo experiments with 100 independent trials 
per configuration. We fix the number of runways $R = 2$, service 
rates $[\mu_1, \mu_2] = [2, 2]$ aircraft per epoch, and initial 
queue lengths $q = [3, 4]$. In each trial, a busy-airport instance 
is generated by randomly drawing scheduled lateness 
$\ell_f \sim \mathcal{N}(0, 10^2)$ truncated to $\ell_f \geq 0$ 
for each eligible flight. The experimental variables are the 
eligible aircraft count $|\mathcal{F}| \in \{6, 7, \ldots, 14\}$, 
the confidence level $\alpha \in [0.30, 0.99]$, and the 
uncertainty level $\sigma_i$. The aircraft count $|\mathcal{F}| = 14$ 
corresponds to approximately 210 eligible pushbacks per hour,
assuming a 4-minute decision epoch,
exceeding the departure rate of one of the world's busiest 
airports~\cite{example_atlanta}. Each flight is randomly assigned 
to one of five airlines.

We compare four coordination mechanisms:
\begin{itemize}
    \item \texttt{FCFS}: Decentralized first-come-first-served, 
    reflecting current VQ practice~\cite{VQ_asdex, VQ_burgain}. 
    Airlines independently select aircraft by earliest scheduled 
    departure.
    \item \texttt{Full-CCCE}: Full CC-CE solver via 
    \Cref{eq:ccce_program}.
    \item \texttt{RR-Nominal}: Reduced-rank CE solver via 
    \Cref{eq:rr_program} under nominal costs ($\sigma = 0$), 
    without uncertainty awareness.
    \item \texttt{RR-CCCE}: Reduced-rank CC-CE solver via 
    \Cref{eq:rr_program} with uncertainty awareness.
\end{itemize}
\Cref{tab:baseline} summarizes the key differences among the 
compared coordination mechanisms.

\begin{table}[h]
\centering
\caption{Comparison of coordination mechanisms}
\label{tab:baseline}
\begin{tabular}{lcc}
\toprule
Method & CE formulation & Uncertainty-aware \\
\midrule
\texttt{FCFS} & -- & -- \\
\texttt{Full-CCCE} & Full formulation~\Cref{eq:ccce_program} & Yes \\
\texttt{RR-Nominal} & Reduced formulation~\Cref{eq:rr_program} & No \\
\texttt{RR-CCCE} & Reduced formulation~\Cref{eq:rr_program} & Yes \\
\bottomrule
\end{tabular}
\end{table}

We evaluate solver performance using the following three metrics.
\begin{itemize}
    \item \textbf{Delay cost}: The realized total delay 
    cost~\eqref{eq:sys_cost} computed under the joint pushback 
    decision. This metric captures both schedule delay and 
    congestion-induced delay.
    \item \textbf{Computation time}: Wall-clock time required 
    to compute the equilibrium at each trial. We use the 
    real-time constraint of 4 minutes per epoch as the 
    deployment threshold for online applicability.
    \item \textbf{Deviation rate}: The fraction of trials in 
    which at least one airline deviates from the recommended 
    action. A lower deviation rate indicates stronger incentive 
    compatibility under cost uncertainty.
\end{itemize}

For RRCE-based methods, CC-PNEs are identified by exhaustive enumeration of pure strategy profiles, which remains tractable for the instances considered in this work.
Experiments were implemented in Julia v1.11.3 using ParametricMCP~\cite{parammcp} and PATHSolver~\cite{PATH}.

\subsection{Result: Scalability}

We first evaluate computational scalability as the number of agents increases.
\texttt{Full-CCCE} computation time grows rapidly as the number 
of eligible aircraft increases, due to the exponential expansion 
of the joint action space, exceeding the 4-minute real-time 
threshold at 9 aircraft per epoch as shown in 
\Cref{fig:scalability}. \texttt{RR-CCCE} remains well below 
this threshold across all tested traffic levels up to 
$|\mathcal{F}| = 14$, corresponding to approximately 210 
eligible pushbacks per hour. This demonstrates that the 
reduced-rank approach enables real-time deployment at realistic coordination scenarios.

\begin{figure}[t!]
    \vspace{2mm}
    \centering
    \includegraphics[width=0.99\linewidth]{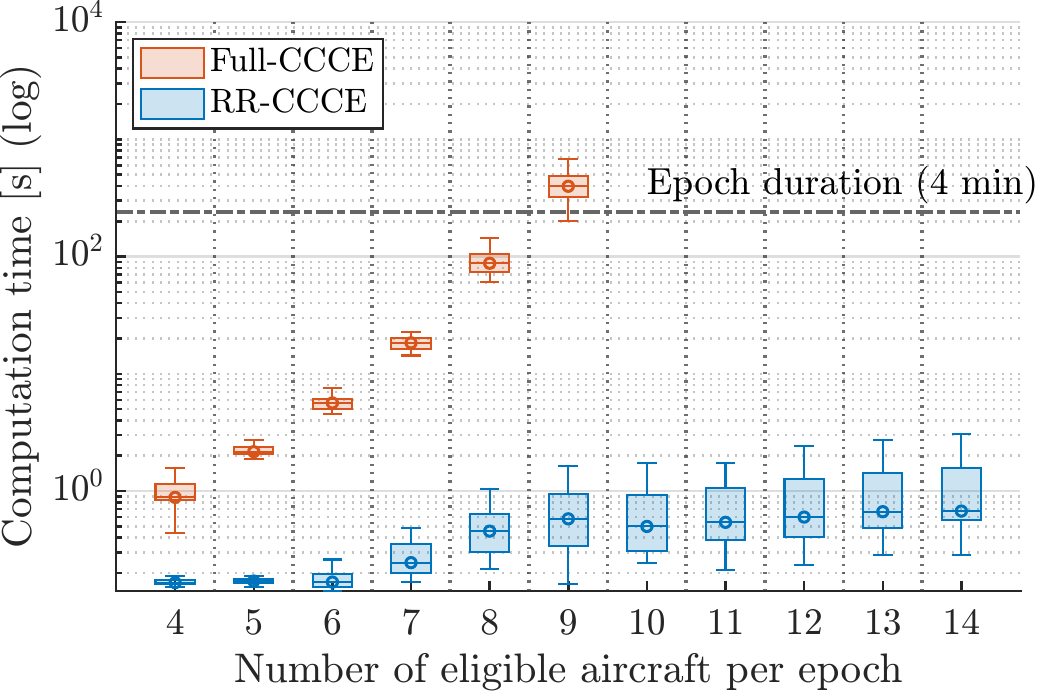}
    \caption{Scalability comparison between \texttt{Full-CCCE} and \texttt{RR-CCCE}.
    Wall-clock computation time (log scale) is shown as a function of the number of eligible aircraft per epoch.
    The horizontal dotted line indicates the epoch duration (4 minutes), representing the real-time constraint for online deployment.}
    \label{fig:scalability}
\end{figure}

\begin{figure}[t!]
    \centering
    \includegraphics[width=0.99\linewidth]{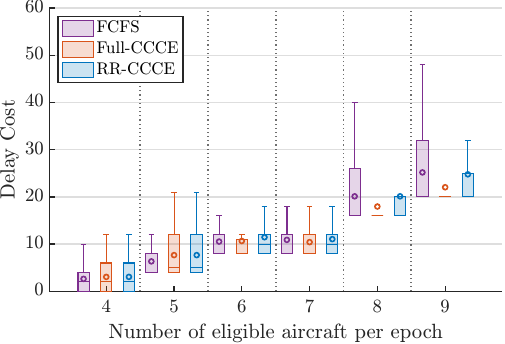}
    \caption{Comparison of realized delay cost across coordination mechanisms.
    \texttt{Full-CCCE} and \texttt{RR-CCCE} outperform \texttt{FCFS} as traffic increases.
    \texttt{RR-CCCE} exhibits a performance gap relative to \texttt{Full-CCCE} due to approximation. Average marked by $\circ$.}
    \label{fig:delayCost}
\end{figure}

\subsection{Result: Coordination effect}

We next examine the coordination performance under nominal costs.
The realized delay cost across coordination mechanisms under 
$\sigma = 0$ is shown in \Cref{fig:delayCost}. At low traffic 
levels, all methods perform similarly, though \texttt{Full-CCCE} 
and \texttt{RR-CCCE} show higher variance. As traffic increases, 
\texttt{FCFS} cost rises sharply, while CE-based methods yield 
lower costs. In the high-volume regime ($|\mathcal{F}| \in 
\{8, 9\}$), \texttt{Full-CCCE} achieves the lowest delay cost compared
to the others. \texttt{RR-CCCE} 
exhibits a modest performance gap relative to 
\texttt{Full-CCCE}, attributable to the reduced-rank 
approximation of the equilibrium set. These results confirm 
that CE-based coordination consistently improves system-level 
performance over the current operational baseline.

\begin{figure}[t]
    \vspace{2mm}
  \centering
    \begin{subfigure}{0.44\textwidth} \label{fig:eff_sigma}
        \includegraphics[width=\linewidth]{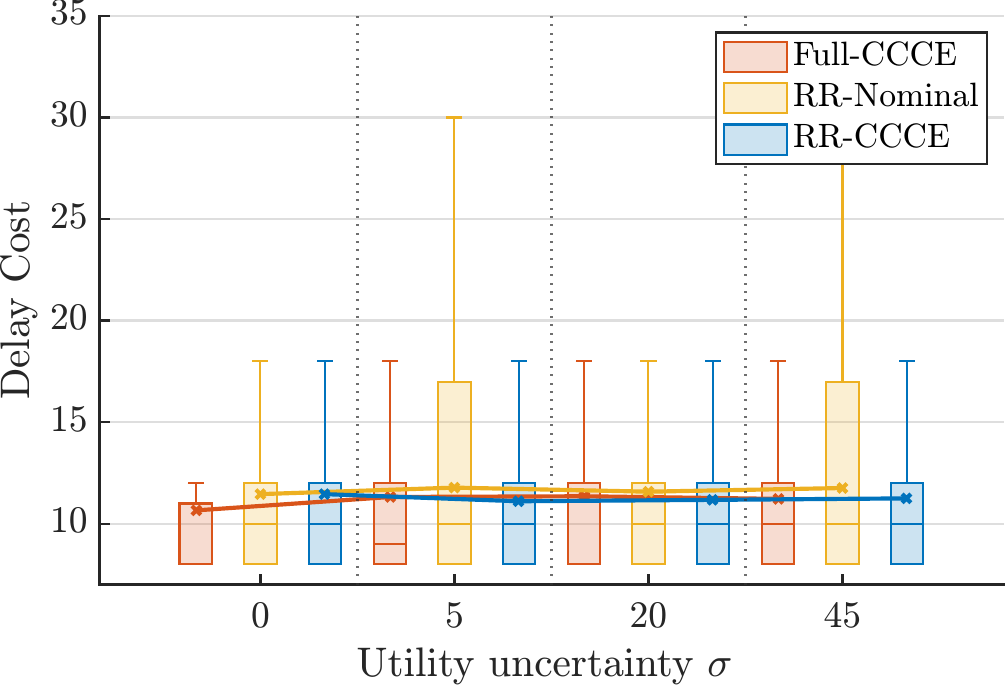}
        \caption{}
    \end{subfigure}
    \begin{subfigure}{0.44\textwidth} \label{fig:devfreq_sigma}
        \includegraphics[width=\linewidth]{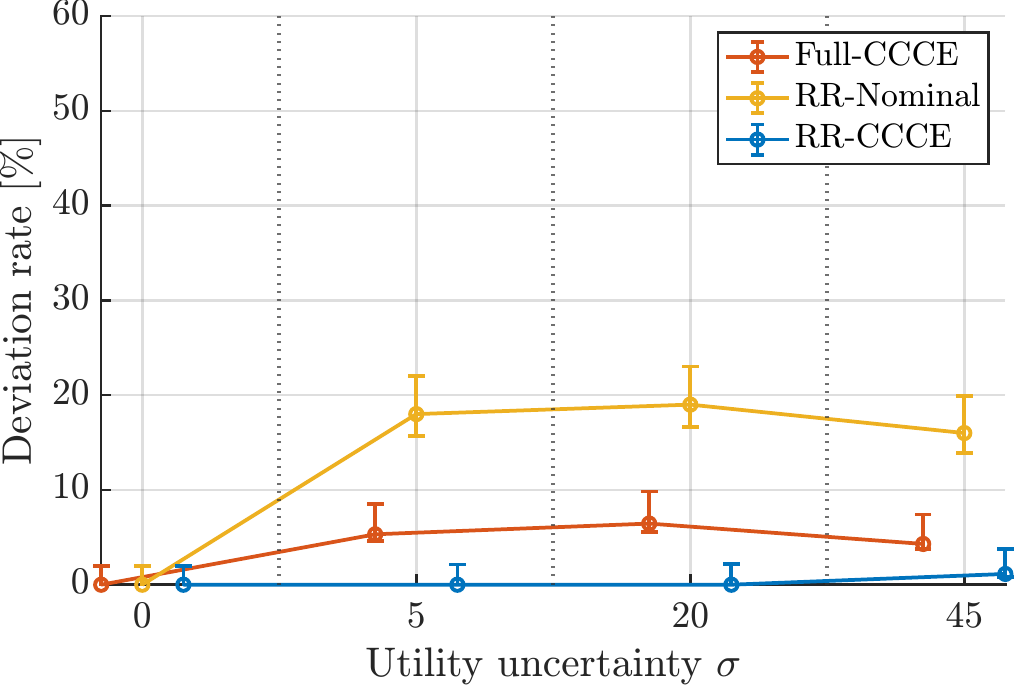}
        \caption{}
    \end{subfigure}
  \caption{Performance under cost uncertainty ($\alpha=90\%$, six airlines).
    (a) Delay cost distribution as a function of cost noise level $\sigma$. 
    All correlated equilibrium-based methods maintain comparable delay cost as uncertainty increases.
    (b) Deviation rate versus $\sigma$. 
    The uncertainty-aware formulation reduces deviation rates, with the reduced-rank approach maintaining consistently low deviation levels across all $\sigma$.
    Horizontal lines indicate mean values.}
  \label{fig:sigma_2stack}
\end{figure}

\subsection{Result: Robustness under uncertainty}

Finally, we investigate the impact of cost uncertainty.
Under varying cost uncertainty $\sigma$ at $\alpha = 90\%$, 
mean delay costs remain comparable across all CE-based methods 
as shown in \Cref{fig:sigma_2stack}(a), indicating that 
coordination performance is largely preserved under uncertainty. 
Regarding deviation frequency, uncertainty-aware methods 
demonstrate improved robustness as shown in 
\Cref{fig:sigma_2stack}(b). \texttt{RR-Nominal} exhibits 
consistently elevated deviation rates of 10--20\% across all 
$\sigma$ values, while \texttt{Full-CCCE} shows moderate 
deviation rates that increase with $\sigma$. \texttt{RR-CCCE} 
maintains consistently lower deviation rates across all tested 
uncertainty levels. Notably, as $\sigma$ increases, the delay 
cost gap between \texttt{Full-CCCE} and \texttt{RR-CCCE} 
narrows, suggesting that uncertainty-driven contraction of the 
full CC-CE feasible set reduces the advantage of the full 
formulation. These results indicate that \texttt{RR-CCCE} 
achieves maintains low deviation rates without sacrificing 
coordination performance.

\section{Discussion}
\label{sec:discussion}

The results show that correlated equilibrium-based coordination 
consistently reduces delay cost compared to the \texttt{FCFS} 
baseline across all tested traffic levels. This indicates that 
incentive-compatible recommendations can improve system-level 
performance without requiring centralized enforcement. 
\texttt{RR-CCCE} scales to realistic traffic levels within the 
real-time constraint, whereas \texttt{Full-CCCE} becomes 
computationally intractable beyond moderate problem sizes. 
Under nominal costs, \texttt{RR-CCCE} exhibits a modest 
performance gap relative to \texttt{Full-CCCE}, reflecting 
the restriction imposed by the reduced-rank approximation.

Under cost uncertainty, two trends are observed. First, the 
delay cost gap between \texttt{Full-CCCE} and \texttt{RR-CCCE} 
largely diminishes, while \texttt{RR-Nominal} shows increased 
variance, indicating the importance of explicitly incorporating 
uncertainty in the equilibrium constraints. Second, 
\texttt{RR-CCCE} maintains consistently low deviation rates 
across all tested uncertainty levels, while \texttt{Full-CCCE} 
exhibits higher deviation rates as uncertainty increases. 
This suggests that the restricted feasible set induced by the reduced-rank formulation may lead to an incidental robustness effect.

Overall, these results indicate that \texttt{RR-CCCE} achieves 
a favorable tradeoff between computational efficiency, 
coordination performance, and robustness under uncertainty, 
making it a practical approach for large-scale noncooperative coordination problems.

\section{Conclusion}

We developed a scalable method for computing chance-constrained correlated equilibria (CC-CE) by extending the reduced-rank correlated equilibrium framework to the chance-constrained setting.
By approximating the CC-CE feasible set via convex combinations of chance-constrained pure Nash equilibria, the coordination problem reduces from exponential in the number of agents to a small linear program. 
Numerical experiments on a collaborative virtual queue scenario confirmed scalability to realistic traffic levels and consistent performance improvements over current practice. 
Under cost uncertainty, the reduced-rank approach maintained lower deviation rates while achieving coordination performance comparable to the full formulation, indicating a potential robustness benefit of the reduced-rank approximation. 
Future work includes extending the framework to dynamic multi-epoch settings and developing computational approaches that account for the possible nonexistence of chance-constrained pure Nash equilibria.

\bibliographystyle{IEEEtran} 
\bibliography{reference}

\end{document}